\documentclass[aps, pra, twocolumn, superscriptaddress]{revtex4-2}
\usepackage[utf8]{inputenc}
\usepackage{amsmath, amsthm, amssymb}
\usepackage{amsthm}
\usepackage{amsfonts}
\usepackage{tikz}
\usepackage{braket}
\usepackage{graphicx}
\usepackage[english]{babel}
\usepackage{comment}
\usepackage[colorlinks=true,linkcolor=blue,urlcolor=blue,citecolor=blue]{hyperref}
\newtheorem{prop}{Proposition}
\newtheorem{theorem}{Theorem}

\begin{document}
\title{Information locking and its resource efficient extraction}

\author{Suchetana Goswami}
\email{suchetana.goswami@gmail.com}
\affiliation{Centre for Quantum Optical Technologies, Centre of New Technologies, University of Warsaw, Banacha 2c, 02-097 Warsaw, Poland}
\affiliation{Harish-Chandra Research Institute, A CI of Homi Bhabha National Institute, Chhatnag Road, Jhunsi, Allahabad 211 019, India}

\author{Saronath Halder}
\email{saronath.halder@gmail.com}
\affiliation{Centre for Quantum Optical Technologies, Centre of New Technologies, University of Warsaw, Banacha 2c, 02-097 Warsaw, Poland}

\begin{abstract}

Locally indistinguishable states are useful to distribute information among spatially separated parties such that the information is locked. This implies that the parties are not able to extract the information completely via local operations and classical communication (LOCC) while it might be possible via LOCC when the parties share entanglement. In this work, we consider an information distribution protocol using orthogonal states for $m\geq3$ spatially separated parties such that even if any $k\leq (m-1)$ parties collaborate still the information cannot be revealed completely. Such a protocol is useful to understand up to what extent the encoded information remains locked. However, if required, the parties can share entanglement and extract the information completely by LOCC. To make the process resource efficient, it should consume less number of entangled states. We show that though the set of states, which are locally indistinguishable across every bipartition, are sufficient for the above protocol, they may consume higher number of entangled states when aiming for complete information extraction. We establish this by constructing a class of locally indistinguishable sets of orthogonal states which can be employed to accomplish the above protocol and these sets consume less number of entangled states, compared to the former sets, for complete information extraction. In fact, this difference in the number of required entangled states for complete information extraction grows linearly with the number of parties. This study sheds light on suitable use of local indistinguishability property of quantum states as resource and thus, we demonstrate an efficient way of information distribution.

\end{abstract}
\maketitle

\section{Introduction}
Distinguishing quantum states \cite{Barnett09, Chefles00, Bergou10, Bae15} is one of the key steps in many information processing protocols. Such a step can be thought of in the following way. Suppose, a quantum system is prepared in an unknown state. But the state is taken from a known set. The goal is to identify the state of the quantum system. If the states of the known set are pairwise orthogonal to each other then, in principle, it is possible to identify the state of the system perfectly by performing an appropriate measurement on the whole system. On the other hand, nonorthogonal states cannot be distinguished perfectly \cite{Nielsen00}.

We assume that a composite quantum system is distributed among several spatially separated parties and the parties are restricted to perform local quantum operations and classical communication (LOCC) only. In such a situation, it may not always be possible to identify the state of the system perfectly even though, the states of the known set are orthogonal to each other \cite{Bennett99-1, Ghosh01, Walgate02, Ghosh02, Horodecki03, Fan04, Ghosh04, Watrous05, Hayashi06, Bandyopadhyay11, Yu12, Halder18, Halder19}. For a given set, if it is not possible to identify the state of the system perfectly then the set is said to be a locally indistinguishable set, otherwise, the set is distinguishable. Locally indistinguishable sets find applications in data hiding \cite{Terhal01, Eggeling02, Lami21, Bandyopadhyay21}, secret sharing \cite{Markham08, Rahaman15}, etc.

In this work, we consider an information distribution task and ask which type of locally indistinguishable sets are appropriate to complete the task. In this context, we keep in our mind that local indistinguishability of quantum states is a resource and one should use it suitably. Anyway, the task can be described in the following manner. Suppose, there is a Referee who wishes to distribute an $N$-level classical information among $m$ spatially separated parties, $N>2$ and $m\geq3$. But this should be done in such a way that even if, certain number of parties $k$, $2\leq k\leq(m-1)$, collaborate, the information is not revealed completely. These collaborating parties are allowed to perform joint measurements on their subsystems and the rest of the parties stand alone, i.e., they are only allowed to perform measurements on their own subsystems. But to make strategies, any sequence of classical communication is allowed among the parties. However, if required, then there must be a way such that the parties can extract the information completely by sharing entangled states as resource among them along with LOCC. Now, sharing entanglement among spatially separated parties is always a difficult job to implement. Therefore, the referee should try to accomplish the task in a way that consumption of entangled states can be reduced for complete information extraction when it is required. Here comes the role of suitably using local indistinguishability of quantum states as resource. 

We note that one way, to implement the collaboration among some parties, is to allow them sharing classical communication (CC) while the non-collaborating parties do not use CC. To beat any such collaboration, here we allow the collaborating parties to use joint measurements. Thus, we basically are in search of robust information distribution protocol. Clearly, such a protocol is also useful to understand up to what extent the privacy of the encoded information remains intact.

Implementing the above task might be easier if we drop the condition that one has to reduce the consumption of entangled states for complete information extraction (we say this condition as `resource-efficient' condition). A quick solution is given as the following. We consider a set of $N$ orthogonal pure $m$-partite states. The classical information is encoded against the states of the set. We also assume that the set is locally indistinguishable across every bipartition. (For sets which have local indistinguishability across bipartition(s), one can go through the Refs.~\cite{Halder19, Rout19, Zhang19, Halder20-1, Yuan20, Shi20, Rout21, Shi21, Wang21, Li21, Shi22} and the references therein.) Such a set is always sufficient for the implementation of the present task. Here the orthogonality is to preserve the condition that there must be a way for complete information extraction when it is required. Now, given a set of orthogonal quantum states, if the states cannot be perfectly distinguished by LOCC across every bipartition, then these states must also not be perfectly distinguished by LOCC in any multipartition. So, for such a set, does not matter how many parties are collaborating, in the newly produced partition, the set always remains locally indistinguishable and thus, the information, encoded against the states of the set, cannot be extracted completely. Probably, we are now ready to rephrase the main question which is addressed in this work: Is it possible to find more suitable sets to implement the present task compared to the sets which are locally indistinguishable across every bipartition? This question is particularly important when we do not drop the resource-efficient condition.

The answer to the above question is not obvious. In fact, when the number of parties is three, the sets which can be used to accomplish the present task, are indeed locally indistinguishable across every bipartition. This can be understood in the following way. Suppose, there are three parties $A$, $B$, and $C$. Then, in this case the only value of $k$ is 2. So, if any two of the three parties collaborate, then the partitions, which are produced due to collaboration, are $A-BC$, $B-AC$, and $C-AB$. Again, in a tripartite system these are the only possible bipartitions. Therefore, the tripartite sets of orthogonal states which are locally indistinguishable across the aforesaid bipartitions, are indeed locally indistinguishable across every bipartition. 

Nevertheless, when the number of parties increases, i.e., $m\geq4$, it is possible to show that there are sets which are not only sufficient to accomplish the present task but they may also be resource-efficient compared to the sets which are locally indistinguishable across every bipartition. For the construction of the present sets, we use pairwise orthogonal Greenberger–Horne–Zeilinger (GHZ) type states \cite{Greenberger07} (for distinguishability of GHZ basis, see \cite{Bandyopadhyay18}). We mention that here within a set, the states are pure and they are equally probable. We also mention that the entangled states, which are available as resource, are two-qubit maximally entangled states, which can be shared between two parties. 

The main contributions of this paper is given as the following: (i) We construct a class of sets which contains maximally entangled multi-qubit GHZ states. These sets are locally indistinguishable across some bipartitions but not in every bipartitions. Again, these sets are sufficient to accomplish the present task for certain values of $N$. (ii) We show that this sets can be more resource-efficient than the sets which are locally indistinguishable across every bipartition. (iii) We define a quantity $\Delta E$ as the difference in the number of entangled states which are consumed for complete information extraction in case of the present sets and the sets which are locally indistinguishable across every bipartition. We also show that $\Delta E$ increases with increasing $m$, $m\geq4$ and $m$ is either even or odd.  

Due to above findings, a few things are clear now. If for a given $m$-partite ($m\geq4$) set of orthogonal states, any $(m-1)$ parties collaborate and they are not able to extract the information completely, then it does not mean that all parties have to collaborate for complete information extraction. This fact can be utilized in an information processing protocol. In fact, equivalently, for the present protocol, it is not necessary to use an $m$-partite ($m\geq4$) set which is locally indistinguishable across every bipartition. Our task and corresponding examples also exhibit instances where more local indistinguishability cannot guarantee more efficiency.


\section{Results}\label{sec2}
We consider $m$-partite system where each party holds only one qubit. To encode $N$-level classical information, one needs a set of $N$ quantum states. Therefore, the cardinality of the considered set is $N$. In fact, $N$ changes with increasing $m$ as $N=m+2$ in our case. We also mention that here we consider only orthogonal pure states and perfect discrimination of these states is considered.

\subsection{Four-qubit case}
We consider a four-partite qubit system ($\mathcal{C}^2 \otimes \mathcal{C}^2 \otimes \mathcal{C}^2 \otimes \mathcal{C}^2$) shared between four parties, $A_1$, $A_2$, $A_3$, and $A_4$. Let us construct a set (say, $S^4_1$) of four-qubit states 
which only contains pure maximally entangled GHZ states. The form of the set is given below.
\begin{eqnarray}
S^4_1: \{\ket{0000} \pm \ket{1111}, \nonumber\\
\ket{1000} + \ket{0111}, \nonumber\\
\ket{0100} + \ket{1011}, \nonumber\\
\ket{0010} + \ket{1101}, \nonumber\\
\ket{0001} + \ket{1110}\}
\label{s1_4}
\end{eqnarray}
Note that, for simplicity we do not consider the normalisation factors. These factors do not have any relevance in the discrimination process. For now on, we can use the notation $d\otimes d^\prime$ instead of $\mathcal{C}^d\otimes\mathcal{C}^{d^\prime}$. For $S^4_1$, the value of $N$ is six.  

\begin{prop}\label{prop1}
$S^4_1$ is locally indistinguishable across every $2 \otimes 2^3$ bipartition but locally distinguishable across every $4 \otimes 4$ bipartition. 
\end{prop}

\begin{proof}
Let us first consider the bipartition as $A_1 - A_2 A_3 A_4$. We denote the 3-qubit basis of $A_2 A_3 A_4$ (or $A'_3$) as $\{\ket{i'_3}\}_{i=0}^7$, where $\ket{0'_3}\equiv \ket{000}$, $\ket{1'_3}\equiv \ket{001}$ and so on. Hence the states in $S^4_1$ can be rewritten as \{$(\ket{0}\ket{0'_3}\pm \ket{1}\ket{7'_3})$, $(\ket{1}\ket{0'_3}+\ket{0}\ket{7'_3})$, $(\ket{0} \ket{4'_3}+\ket{1} \ket{3'_3})$, $(\ket{0} \ket{2'_3} + \ket{1}\ket{5'_3})$, $(\ket{0} \ket{1'_3}+\ket{1}\ket{6'_3})$\}. Now, the side $A'_3$ performs a measurement on the three qubits. Note that while it is possible to distinguish the last three states by LOCC, it is impossible to distinguish among the first three. The reason behind this indistinguishability is that the three states resemble three Bell states of two qubits. Now, it has already been shown in literature that it is not possible to distinguish three or four Bell states perfectly via LOCC \cite{Ghosh01}. Following similar technique, if we consider any $2\otimes 2^3$ bipartition, it is always possible to find three Bell-like indistinguishable states. Thus, the above set cannot be perfectly distinguished by LOCC in these bipartitions.

On the other hand, when we consider the bipartition of the form $A_1 A_2 - A_3 A_4$, we show that it is possible to distinguish the states of $S^4_1$ with local measurements. Here, we denote the 2-qubit basis of the first subsystem $A_1 A_2$ (or, $A^{(1)}_2$) as $\{\ket{j^{(1)}_2}\}_{j=0}^3$ such that $\ket{0^{(1)}_2}\equiv \ket{00}$, $\ket{1^{(1)}_2} \equiv \ket{01}$, and so on. Similarly, for the second subsystem $A_3 A_4$ (or, $A^{(2)}_2$) as $\{\ket{j^{(2)}_2}\}_{j=0}^3$ such that, $\ket{0^{(2)}_2}\equiv \ket{00}$, $\ket{1^{(2)}_2}\equiv \ket{01}$ and so on. Hence, the states in $S^4_1$ can be re-written as, \{$(\ket{0^{(1)}_2}\ket{0^{(2)}_2}\pm \ket{3^{(1)}_2}\ket{3^{(2)}_2})$, $(\ket{2^{(1)}_2}\ket{0^{(2)}_2}+\ket{1^{(1)}_2}\ket{3^{(2)}_2})$, $(\ket{1^{(1)}_2}\ket{0^{(2)}_2}+\ket{2^{(1)}_2}\ket{3^{(2)}_2})$, $(\ket{0^{(1)}_2}\ket{2^{(2)}_2}+\ket{3^{(1)}_2}\ket{1^{(2)}_2})$, $(\ket{0^{(1)}_2}\ket{1^{(2)}_2}+\ket{3^{(1)}_2}\ket{2^{(2)}_2})$\}. Notice that when $A^{(2)}_2$ performs the projective measurements, where the projectors are $(\ket{0^{(2)}_2}\bra{0^{(2)}_2}+\ket{3^{(2)}_2}\bra{3^{(2)}_2})$ and $(\ket{1^{(2)}_2}\bra{1^{(2)}_2}+\ket{2^{(2)}_2}\bra{2^{(2)}_2})$, it is possible to distinguish between the subspaces, spanned by the first four states and the last two. Now, for the last two states, being orthogonal pure states, they are always locally distinguishable \cite{Walgate00}. On the other hand, for the first four states, one can consider projective measurement on $A^{(1)}_2$, where the projectors are given by $(\ket{0^{(1)}_2}\bra{0^{(1)}_2}+\ket{3^{(1)}_2}\bra{3^{(1)}_2})$ and $(\ket{1^{(1)}_2}\bra{1^{(1)}_2}+\ket{2^{(1)}_2}\bra{2^{(1)}_2})$. This is to separate out the subspaces, spanned by the first two and last two states. Finally, after subspace discrimination, only two orthogonal pure states are left, which can be distinguished by LOCC \cite{Walgate00}. This analysis also holds for other $4\otimes 4$ bipartition. These complete the proof. 
\end{proof}

Here we consider four parties: $A_1$, $A_2$, $A_3$, and $A_4$. We suppose that $k$ parties among them collaborate then either $k = 2$ or $k=3$. If $k=3$ then the possible bipartitions are $A_1-A_2A_3A_4$, $A_2-A_1A_3A_4$, $A_3-A_1A_2A_4$, and $A_4-A_1A_2A_3$. Again, if $k=2$, then the possible partitions are tripartitions which are given by $A_1A_2-A_3-A_4$, $A_1A_3-A_2-A_4$, $A_1A_4-A_2-A_3$, $A_1-A_2A_3-A_4$, $A_1-A_2A_4-A_3$, and $A_1-A_2-A_3A_4$. Clearly, we can encode the information against the four-qubit states of a set which is locally indistinguishable across all $2\otimes8$ bipartitions. Such a set can be found from Eq.~(\ref{s1_4}). Now, we want to think about the complete information extraction part. The set $S^4_1$ is distinguishable across $A_1A_2-A_3A_4$ bipartition. In this case, if the parties $A_1, A_2$ and $A_3, A_4$ share two-qubit pure maximally entangled states (Bell states), then it is sufficient to locally distinguish the states perfectly. It is due to a teleportation \cite{Bennett93} based protocol. $A_1$ can teleport the qubit to the location of $A_2$. Similarly, $A_3$ can teleport the qubit to the location of $A_4$. In this way, the bipartition $A_1A_2-A_3A_4$ is produced.

On the other hand, it is already explained that given any set which is locally indistinguishable across every bipartition, are also sufficient to accomplish the present task. Now, for four qubits, two entangled states cannot be sufficient for complete information extraction using such sets. Because for a four-qubit set $S^4_2$ which is locally indistinguishable across every bipartition, if one uses two bipartite maximally entangled states and follow a teleportation based protocol, then ultimately, a new bipartition will be produced in which $S^4_2$ is again, locally indistinguishable. So, entangled states required for complete information extraction when the set is $S^4_1$, given by $E(S^4_1)$ = 2 and similarly, $E(S^4_2)$ = 3 (at least necessary). Thus, $\Delta E = E(S^4_2)-E(S^4_1) = 1$.



\subsection{Six-qubit case}
Now we consider the case consisting of six-qubit ($\mathcal{C}^2 \otimes \mathcal{C}^2 \otimes \mathcal{C}^2 \otimes \mathcal{C}^2 \otimes \mathcal{C}^2 \otimes \mathcal{C}^2$) shared between six parties, $A_1$, $A_2$, $A_3$, $A_4$, $A_5$ and $A_6$. We construct a set ($S^6_1$) of six-qubit state 
which only contains pure maximally entangled GHZ states. The form of the set is given below. 
\begin{eqnarray}
S^6_1: \{\ket{000000} \pm \ket{111111}, \nonumber\\
\ket{100000} + \ket{011111}, \nonumber\\
\ket{010000} + \ket{101111}, \nonumber\\
\ket{001000} + \ket{110111}, \nonumber\\
\ket{000100} + \ket{111011}, \nonumber\\
\ket{000010} + \ket{111101}, \nonumber\\
\ket{000001} + \ket{111110}\}
\label{s1_6}
\end{eqnarray}
For simplicity we discard the normalisation as before because it does not play any important role in our protocol.
\begin{prop}
$S^6_1$ is locally indistinguishable across every $2 \otimes 2^5$ bipartition but locally distinguishable across every $4 \otimes 4 \otimes 4$ tripartition. 
\end{prop}

\begin{proof}
Following the similar mechanism as used in the proof of the previous proposition, we consider first the bipartition as $A_1 - A_2 A_3 A_4 A_5 A_6$. We then define the 5-qubit basis of $A_2 A_3 A_4 A_5 A_6$ (or, $A'_5$) as $\{\ket{i'_5}\}^{31}_{i=0}$ where, $\ket{0'_5}\equiv\ket{00000}$, $\ket{1'_5}\equiv\ket{00001}$,$\cdots$, $\ket{31'_5}\equiv\ket{11111}$. Hence, the set $S^6_1$ can be rewritten as, \{$(\ket{0}\ket{0'_5} \pm \ket{1}\ket{31'_5})$, $(\ket{1}\ket{0'_5} + \ket{0}\ket{31'_5})$, $(\ket{0}\ket{16'_5} + \ket{1}\ket{15'_5})$, $(\ket{0}\ket{8'_5} + \ket{1}\ket{23'_5})$, $(\ket{0}\ket{4'_5} + \ket{1}\ket{27'_5})$, $(\ket{0}\ket{2'_5} + \ket{1}\ket{29'_5})$, $(\ket{0}\ket{1'_5} + \ket{1}\ket{30'_5})$\}. Note that if on the second composite system i.e., on $A'_5$ projective measurements are performed then the last five states (last five states of the above equation) can be distinguished. But the first three states can be seen as three Bell states in the bipartite system $A_1-A'_5$ which cannot be perfectly distinguished by LOCC \cite{Ghosh01}. Following similar technique, if we consider any $2\otimes 2^5$ bipartition, it is always possible to find three Bell-like indistinguishable states. Thus, the above set cannot be perfectly distinguished by LOCC in these bipartitions.

On the other hand, while considering the tripartition, the subsystems can be grouped as $A_1 A_2$ (or, $A^{(1)}_2$), $A_3 A_4$ (or, $A^{(2)}_2$) and $A_5 A_6$ (or, $A^{(3)}_2$). We define the new basis of the corresponding subsystems as, $\{j^{(1)}_2\}_{j=0}^3$, $\{j^{(2)}_2\}_{j=0}^3$, and $\{j^{(3)}_2\}_{j=0}^3$ respectively (as defined in the proof of the Proposition \ref{prop1}). Hence the states in the set $S^6_1$ can be rewritten as, \{($\ket{0^{(1)}_2}\ket{0^{(2)}_2}\ket{0^{(3)}_2}\pm\ket{3^{(1)}_2}\ket{3^{(2)}_2}\ket{3^{(3)}_2}$), ($\ket{2^{(1)}_2}\ket{0^{(2)}_2}\ket{0^{(3)}_2}+\ket{1^{(1)}_2}\ket{3^{(2)}_2}\ket{3^{(3)}_2}$), ($\ket{1^{(1)}_2}\ket{0^{(2)}_2}\ket{0^{(3)}_2}+\ket{2^{(1)}_2}\ket{3^{(2)}_2}\ket{3^{(3)}_2}$), ($\ket{0^{(1)}_2}\ket{2^{(2)}_2}\ket{0^{(3)}_2}+\ket{3^{(1)}_2}\ket{1^{(2)}_2}\ket{3^{(3)}_2}$), ($\ket{0^{(1)}_2}\ket{1^{(2)}_2}\ket{0^{(3)}_2}+\ket{3^{(1)}_2}\ket{2^{(2)}_2}\ket{3^{(3)}_2}$), ($\ket{0^{(1)}_2}\ket{0^{(2)}_2}\ket{2^{(3)}_2}+\ket{3^{(1)}_2}\ket{3^{(2)}_2}\ket{1^{(3)}_2}$), ($\ket{0^{(1)}_2}\ket{0^{(2)}_2}\ket{1^{(3)}_2}+\ket{3^{(1)}_2}\ket{3^{(2)}_2}\ket{2^{(3)}_2}$)\}. First, the subsystem $A^{(3)}_2$ performs projective measurement with projectors given as, $(\ket{0^{(3)}_2}\bra{0^{(3)}_2}+\ket{3^{(3)}_2}\bra{3^{(3)}_2})$ and $(\ket{1^{(3)}_2}\bra{1^{(3)}_2}+\ket{2^{(3)}_2}\bra{2^{(3)}_2})$ revealing the subspaces consisting the first six states and the last two. Note that, the last two states can be seen as a pair of orthogonal maximally entangled states in the newly defined basis for the grouped susbsystems and hence can be distinguished via LOCC \cite{Walgate00}. Similarly, for first six states when the subsystem $A^{(2)}_2$ performs the projective measurement with projectors $(\ket{0^{(2)}_2}\bra{0^{(2)}_2}+\ket{3^{(2)}_2}\bra{3^{(2)}_2})$ and $(\ket{1^{(2)}_2}\bra{1^{(2)}_2}+\ket{2^{(2)}_2}\bra{2^{(2)}_2})$ to separate out between the first four and last two states, the last two are again distinguishable by performing LOCC \cite{Walgate00}. Following similar logic when we consider only first four states, on the first subsystem $A^{(1)}_2$ a suitable projective measurement can be performed. The corresponding projectors are $(\ket{0^{(1)}_2}\bra{0^{(1)}_2}+\ket{3^{(1)}_2}\bra{3^{(1)}_2})$ and $(\ket{1^{(1)}_2}\bra{1^{(1)}_2}+\ket{2^{(1)}_2}\bra{2^{(1)}_2})$. Then, it is possible to distinguish between the first two and the other two states of the first four states. Note that the two pure states in the groups are mutually orthogonal to each other and hence can be distinguished perfectly via LOCC \cite{Walgate00}. This analysis also holds for other $4\otimes 4 \otimes 4$ tripartitions. These suffice to prove the proposition. 
\end{proof}

In this case we consider a six qubit system consisting of parties $A_1$, $A_2$, $A_3$, $A_4$, $A_5$ and $A_6$. As can be seen from the above proof, the set of states $S^6_1$ in Eq.~(\ref{s1_6}) cannot be distinguished locally in any $2\otimes 2^5$ bipartition. As a result of which if any $k$ parties collaborate and $(m-k)$ parties stand alone, in the produced bipartition the set remains locally indistinguishable. On the other hand, it can be locally distinguished in the tripartition $A_1 A_2$-$A_3 A_4$-$A_5 A_6$. Note that, following the similar logic as of the four-qubit system in this case, the subsystems can be grouped to produce any $4\otimes 4\otimes 4$ bipartition and for this it is required to share three bipartite maximally entangled states. This is to reveal an eight level information perfectly. Hence, in this case, $E(S^6_1)$ = 3 for complete information extraction. On the other hand, when we have a set (say, $S^6_2$) which is locally indistinguishable across every bipartition, then following the similar logic as in the case of four qubits, we have $E(S^6_2)$ = 5 for revealing the information perfectly. Hence, for six qubits we have, $\Delta E = E(S^6_2)-E(S^6_1) = 2$. Notice that the difference $\Delta E$ is increased as the number of qubits is increased from four to six.

\subsection{Generalisation to $m$-qubit case}

In this section we try to generalise the above findings for an $m$-qubit (considering $m$ to be even) $(\mathcal{C}_1^2 \otimes \mathcal{C}_2^2\otimes\cdots\otimes \mathcal{C}_m^2)$ system shared between parties $\{A_i\}_{i=1}^m$. Following the same trajectory as before we construct a set $S_1^m$ of $m$-qubit states and it is given as the following.
\begin{eqnarray}
S^m_1: &\{\ket{0_1 0_2 \cdots 0_m} \pm \ket{1_1 1_2 \cdots 1_m}, \nonumber\\
&\ket{1_1 0_2 \cdots 0_m} + \ket{0_1 1_2 \cdots 1_m}, \nonumber\\
&\ket{0_1 1_2 \cdots 0_m} + \ket{1_1 0_2 \cdots 1_m}, \nonumber\\
&\vdots \nonumber\\
&\ket{0_1 \cdots 1_{m-1} 0_m} + \ket{1_1 \cdots 0_{m-1} 1_m}, \nonumber\\
&\ket{0_1 0_2 \cdots 1_m} + \ket{1_1 1_2 \cdots 0_m}\}
\label{s1_m}
\end{eqnarray}
For simplicity we discard the normalisation as before as it does not interrupt our findings. Now, we are ready to state the proposition for the general $m$-qubit system. 
\begin{prop}
$S_1^m$ is locally indistinguishable across every $2\otimes2^{(m-1)}$ bipartition but locally distinguishable across every $4\otimes4\otimes \cdots \otimes 4$ $m/2$-partition.
\label{prop_m}
\end{prop}
\begin{proof}
The sketch of the proof relies on the good old method of mathematical induction. First we consider the bipartition as, $A_1 - \{A_i\}_{i=2}^m$. We define the $(m-1)$-qubit basis of $\{A_i\}_{i=2}^m$ (or, $A'_{m-1}$) as $\{\ket{i'_{m-1}}\}_{i=0}^{(2^{m-1}-1)}$. Therefore the states in the set $S_1^m$ can be rewritten as, $\{(\ket{0_1}\ket{0'_{m-1}} \pm \ket{1_1}\ket{(2^{m-1}-1)'_{m-1}}), (\ket{1_1}\ket{0'_{m-1}}+\ket{0_1}\ket{(2^{m-1}-1)'_{m-1}}), \cdots \}$. Note that, the first three states in the set are three orthogonal maximally entangled bipartite states (Bell-like states) while $(m-1)$ parties collaborate between themselves to form the second composite subsystem $A'_{m-1}$. Hence these states can never be distinguished via LOCC \cite{Ghosh01}. Following similar technique, if we consider any $2\otimes 2^{m-1}$ bipartition, it is always possible to find three Bell-like indistinguishable states. Thus, the above set cannot be perfectly distinguished by LOCC in these bipartitions.

Now, on the other hand, we consider the $m/2$-partition ($4\otimes4\otimes\cdots\otimes4$) and see if the states remain locally indistinguishable. For this purpose following the similar technique used in the previous prepositions, we group the subsystems as $A_1 A_2$ (or, $A^{(1)}_2$), $A_3 A_4$ (or, $A^{(2)}_2$), $\cdots$ and $A_{m-1} A_m$ (or, $A^{(m/2)}_2$). Now we define the basis of the newly defined subsystems as $\{\ket{j^{(1)}_2}\}_{j=0}^3$, $\{\ket{j^{(2)}_2}\}_{j=0}^3$, $\cdots$, and $\{\ket{j^{(m/2)}_2}\}_{j=0}^3$ respectively. Hence the states in $S_1^m$ can be rewritten as, \{($\ket{0^{(1)}_2}\ket{0^{(2)}_2} \cdots \ket{0^{(m/2)}_2}\pm\ket{3^{(1)}_2}\ket{3^{(2)}_2} \cdots \ket{3^{(m/2)}_2}$), ($\ket{2^{(1)}_2}\ket{0^{(2)}_2} \cdots \ket{0^{(m/2)}_2}+\ket{1^{(1)}_2}\ket{3^{(2)}_2} \cdots \ket{3^{(m/2)}_2}$), $\cdots$, ($\ket{0^{(1)}_2}\ket{0^{(2)}_2} \cdots \ket{2^{(m/2)}_2}+\ket{3^{(1)}_2}\ket{3^{(2)}_2} \cdots \ket{1^{(m/2)}_2}$), and ($\ket{0^{(1)}_2}\ket{0^{(2)}_2} \cdots \ket{1^{(m/2)}_2}+\ket{3^{(1)}_2}\ket{3^{(2)}_2} \cdots \ket{2^{(m/2)}_2}$)\}. We claim that these states are locally distinguishable in this given partition. Note that, when the subsystem $A_2^{(m/2)}$ performs the projective measurement, given by the projectors, $(\ket{0^{(m/2)}_2}\bra{0^{(m/2)}_2}+\ket{3^{(m/2)}_2}\bra{3^{(m/2)}_2})$ and $(\ket{1^{(m/2)}_2}\bra{1^{(m/2)}_2}+\ket{2^{(m/2)}_2}\bra{2^{(m/2)}_2})$, it separates the last two states in the set and these two states are mutually orthogonal to each other as can be easily seen. Hence, they can be distinguished via LOCC \cite{Walgate00}. Next we consider the subsystem $A_2^{(m/2-1)}$ and it performs the similar projective measurement separating two more mutually orthogonal states, and hence they are locally distinguishable. The procedure can be repeated till the subsystem $A_2^{(1)}$ while finally separates between four remaining states into two sets of a pair of pure orthogonal states which are again locally distinguishable \cite{Walgate00}. This analysis also holds for other $4\otimes 4 \otimes \cdots\otimes4$ (m/2)-partitions. Hence the claim.
\end{proof}

\subsection{Resource efficiency of the task}

Here, we see how the introduced set of states for $m$-qubit system ($m$ being even) $S_1^m$ is more useful than a set of states say, $S_2^m$ which is locally indistinguishable across every bipartition. Note that for a $m$-party system when they have access to the set $S^m_2$, to perform the task efficiently they need to share at least $(m-1)$ bipartite entangled states. The logic is same as for the previously discussed two cases. We assume that the shared states are maximally entangled and they are used in a teleportation based protocol. Then, if the number of such states is less than $(m-1)$ for the set $S^m_2$, it will give rise to a new bipartition along which the set would be locally indistinguishable again. On the other hand, if the set is $S_1^m$ then it is possible to reveal the information completely with less number of maximally entangled states.

\begin{theorem}
The number of bipartite entangled states required for the perfect simulation of the present task with $S_1^m$ is $\frac{m}{2}$ and hence the difference in resource requirement for complete information extraction corresponding to the sets $S_1^m$ and $S_2^m$ grows with the number of parties among which the composite quantum system is distributed.
\end{theorem}
\begin{proof}
As can be seen from Proposition \ref{prop_m} the states in $S_1^m$ are locally indistinguishable across every $1$ vs $(m-1)$ bipartition but locally distinguishable in every $4\otimes4\otimes \cdots \otimes 4$ $m/2$-partition. Hence, if the two parties in the individual subgroups of $m/2$ partitions share a two-qubit maximally entangled state then it is possible to teleport one qubit to the other location. This will enable the set $S_1^m$ to perform the task using only $m/2$ number of bipartite entangled states.

Now, the difference in resource requirement $\Delta E$ can be defined as the following: the number of bipartite entangled states, necessary for complete information extraction from $S^m_2$, i.e., $E(S^m_2)$ difference the number of bipartite entangled states, sufficient for complete information extraction from $S^m_1$, i.e., $E(S^m_1)$.
\begin{eqnarray}
\Delta E &=& E(S^m_2)-E(S^m_1) \nonumber\\
&&=(m-1)-\frac{m}{2} \nonumber\\
&&=\frac{m-2}{2}.
\label{diff}
\end{eqnarray}
Note that for $m=4$, $\Delta E=1$; for $m=6$, $\Delta E=2$ as obtained in the previous individual cases. From Eq. (\ref{diff}) it is clear that as the number of qubits $m$ grows ($m$ is even and $m\geq4$), the task can be made more resource efficiently by using the set of states $S_1^m$.
\end{proof}

\textbf{Remark.--}For odd number of parties, i.e., when $m$ is odd, similar results follow starting from $m=5$ while the set, which is considered, is locally indistinguishable across every $1~vs~(m-1)$ bipartition but locally distinguishable across some $(m-1)/2$- partition. In this case also the difference in resource requirement for complete extraction of information grows with the number of qubits compared to a set of states which is locally distinguishable across every bipartition. 

Therefore, the type of sets, we are talking about, exist in all qubit dimensions when $m\geq4$.


\section{Conclusion}\label{sec3}
In quantum information, when the system in consideration is a composite one, advantages obtained in different tasks are mainly governed by the presence of non-local correlations in the system. Among these correlations quantum entanglement is mostly responsible in speed ups of quantum domain than its classical counterpart \cite{LP_01, V_13, JL_03} but in reality it is an expensive resource. Hence, it is always useful to reduce the use of the same without hindering the effectiveness of the main protocol \cite{V_13, NKGFS_22}. 

In this paper, we have presented an efficient protocol for information sharing such that the information remains locked to a certain extent. Particularly, we have constructed a set of states that are locally indistinguishable across some bipartitions. At the same time, the sets are locally distinguishable across remaining bipartitions. Thus, to decode the information encoded against the states of a present set, we need less number of bipartite entangled states. Hence the set of states prescribed are more useful than the states that are locally indistinguishable across every bipartition when the present task is considered. Because in the later case all the parties need to collaborate together to reveal the information systematically. 

Our construction also depicts the fact that even if $(m-1)$ parties collaborate, the information is not extracted completely -- this does not mean that to extract the information completely all parties have to collaborate. Interestingly, we have noted that the difference in the number of shared bipartite entangled states required in these two types of sets of states, to reveal the information completely, increases linearly with the increasing number of parties. The present task and corresponding examples also exhibit the instances where more local indistinguishability cannot guarantee more efficiency. 

For further research one may consider the following problems: (a) applications of local indistinguishability in information distribution protocols – here one may consider to implement the present task using mixed states or assuming probabilistic setting, (b) understanding the instances like -- more resource may not guarantee more effectiveness, (c) exploring entanglement as resource in complete information extraction, etc.

\section*{Acknowledgements} 
We thank Alexander Streltsov for helpful discussions. This work was supported by the ``Quantum Optical Technologies'' project, carried out within the International Research Agendas programme of the Foundation for Polish Science co-financed by the European Union under the European Regional Development Fund and the ``Quantum Coherence and Entanglement for Quantum Technology'' project, carried out within the First Team programme of the Foundation for Polish Science co-financed by the European Union under the European Regional Development Fund.

\bibliography{ref}
\end{document}